\pdfoutput = 1
\documentclass[12pt, a4paper,  oneside, headinclude,footinclude, BCOR5mm]{scrartcl}

\usepackage[
nochapters, 
beramono, 
eulermath,
pdfspacing, 
dottedtoc 
]{classicthesis} 

\usepackage{bussproofs} 

\usepackage{arsclassica} 

\usepackage[T1]{fontenc} 

\usepackage[utf8]{inputenc} 

\usepackage{graphicx} 
\graphicspath{{Figures/}} 

\usepackage{enumitem} 

\usepackage{lipsum} 

\usepackage{subfig} 

\usepackage{amsmath,amssymb,amsthm} 

\usepackage{varioref} 

\usepackage{tikz}

\newenvironment{bprooftree}
  {\leavevmode\hbox\bgroup}
  {\DisplayProof\egroup}

\def\SP#1{\textsuperscript{{#1}}}

\newtheorem{proposition}{Proposition}[section]
\newtheorem{definition}{Definition}[section]

\title{\normalfont\spacedallcaps{On Computational Paths and the Fundamental Groupoid of a Type}} 
\author{
   Arthur F. Ramos\SP{*}\\
  \texttt{afr@cin.ufpe.br}
  \and
  Ruy J. G. B. de Queiroz\SP{*}\\
  \texttt{ruy@cin.ufpe.br}
\and
  Anjolina G. de Oliveira\SP{*}\\
  \texttt{ago@cin.ufpe.br}
}

\date{} 

\begin{document}

\maketitle 


{\let\thefootnote\relax\footnotetext{* \textit{Centro de Informática, Universidade Federal de Pernambuco, Recife-PE, Brazil}}}


\renewcommand{\sectionmark}[1]{\markright{\spacedlowsmallcaps{#1}}} 
\lehead{\mbox{\llap{\small\thepage\kern1em\color{halfgray} \vline}\color{halfgray}\hspace{0.5em}\rightmark\hfil}} 

\pagestyle{scrheadings} 


\begin{abstract} 

The main objective of this work is to study mathematical properties of computational paths. Originally proposed by de Queiroz \& Gabbay (1994) as `sequences of rewrites', computational paths can be seen as the grounds on which the propositional equality between two computational objects stand. Using computational paths and categorical semantics, we take any type $A$ of type theory and construct a groupoid for this type. We call this groupoid the fundamental groupoid of a type $A$, since it is similar to the one obtained using the homotopical interpretation of the identity type. The main difference is that instead of being just a semantical interpretation, computational paths are entities of the syntax of type theory. We also expand our results, using computational paths to construct fundamental groupoids of higher levels.
\end{abstract}



\section{Introduction}

There seems to be little doubt that the identity type is one of the most intriguing concepts of  Martin-L\"of's Type Theory. This claim is supported by recent groundbreaking discoveries. In 2005, Vladimir Voevodsky \cite{Vlad1} discovered the Univalent Models, resulting in a new area of research known as Homotopy Type Theory \cite{Steve1}. This theory is based on the fact that a term of some identity type, for example $p: Id_{A}(a,b)$, has a clear homotopical interpretation. The interpretation is that the witness $p$ can be seen as a homotopical path between the points $a$ and $b$ within a topological space $A$. This simple interpretation has made clear the connection between Type Theory and Homotopy Theory, generating groundbreaking results, as one can see in \cite{hott,Steve1}. It is important to emphasize that one important fact of the homotopic interpretation is that the homotopic paths exist only in the semantic sense. In other words, there is no formal entity in the language of type theory that represents these paths. They are not present in the syntax of Type Theory.

Given a space $A$, we can think of a structure $\Pi(A)$ formed by points of the space $A$ and homotopical paths between these points. Using straightforward equations, one can easily prove for $\Pi(A)$ that the groupoid equations hold up to homotopy. We call this \emph{weak} structure (it is weak since the equalities do not hold `on the nose', but only up to homotopy) the fundamental groupoid of the space $A$. Since each type $A$, using the homotopical interpretation, can be semantically interpreted as being a space, we can consider $\Pi(A)$ as being the fundamental groupoid of the type $A$.

As we already mentioned, the paths of the homotopical interpretation do not have a counterpart in the syntax of type theory, but exist only as semantical interpretations. With that in mind, adding the concept of paths to the syntax of type theory was the focus of some of our recent works \cite{Ruy1, Art1en, Ruy5}. In these works, we argued that it is possible to formalize the concept of path as an formal entity of the equality theory. This formal entity, known as computational paths, establishes the equality between two terms of the same type. In fact, we argued that the identity type is just the type of this new entity. That way, we added the concept of path to the syntax of type theory, instead of considering it only as a semantic interpretation.

Our main objective in this work is to construct a fundamental groupoid for a type $A$ using the concept of computational path, yet having the same semantics as in the homotopical interpretation. We will also go a little further, showing that it is possible to construct higher groupoid structures, like the weak 2-groupoid of a type $A$ usually written as $\Pi_{2}(A)$.

\section{Computational Paths} \label{path}

As previously mentioned, computational paths will be the main focus here. For that reason, a good understanding of the meaning of this entity is essential. The concept of computational path is inspired on the equality theory of $\lambda$-calculus, known as $\lambda\beta\eta-equality$ \cite{lambda}. This theory estabilishes when two $\lambda$-terms are equal. It establishes a set of axioms and rules of inference that one can use. The idea is that a term $A$ is $\lambda\beta\eta-equal$ to a term $B$ when one can reach $B$ from $A$ after the application of a set of axioms and rules of inference. The application of these axioms and rules of inference generates a path. As a very simple example, since $(\lambda x.x)z \rhd_{\beta} z$, then we say that $(\lambda x.x)z =_{\beta\eta} z$ because of the path $\beta((\lambda x.x)z,z)$. A more complex and interesting example will be showed after we formally define a computational path. Since $\lambda\beta\eta$-equality establishes equalities between terms of the $\lambda$-calculus, we need to find an equivalent theory for Type Theory, which establishes equality between terms of the same type. To do that, we just need to translate to Type Theory the axioms and rules of inference of the $\lambda\beta\eta$-equality. For example, for the product type $\Pi$, we obtain the following axioms (one can check the original axioms of $\lambda\beta\eta-equality$ in \cite{lambda}):

\bigskip

\bigskip

\noindent
\begin{bprooftree}
\hskip -0.3pt
\alwaysNoLine
\AxiomC{$N : A$}
\AxiomC{$[x : A]$}
\UnaryInfC{$M : B$}
\alwaysSingleLine
\LeftLabel{$(\beta$) \ }
\BinaryInfC{$(\lambda x.M)N = M[N/x] : B[N/x]$}
\end{bprooftree}
\begin{bprooftree}
\hskip 11pt
\alwaysNoLine
\AxiomC{$[x : A]$}
\UnaryInfC{$M = M' : B$}
\alwaysSingleLine
\LeftLabel{$(\xi)$ \ }
\UnaryInfC{$\lambda x.M = \lambda x.M' : (\Pi x : A)B$}
\end{bprooftree}

\bigskip

\noindent
\begin{bprooftree}
\hskip -0.5pt
\AxiomC{$M : A$}
\LeftLabel{$(\rho)$ \ }
\UnaryInfC{$M = M : A$}
\end{bprooftree}
\begin{bprooftree}
\hskip 100pt
\AxiomC{$M = M' : A$}
\AxiomC{$N : (\Pi x : A)B$}
\LeftLabel{$(\mu)$ \ }
\BinaryInfC{$NM = NM' : B[M/x]$}
\end{bprooftree}

\bigskip

\noindent
\begin{bprooftree}
\hskip -0.5pt
\AxiomC{$M = N : A$}
\LeftLabel{$(\sigma) \ $}
\UnaryInfC{$N = M : A$}
\end{bprooftree}
\begin{bprooftree}
\hskip 105pt
\AxiomC{$N : A$}
\AxiomC{$M = M' : (\Pi x : A)B$}
\LeftLabel{$(\nu)$ \ }
\BinaryInfC{$MN = M'N : B[N/x]$}
\end{bprooftree}

\bigskip

\noindent
\begin{bprooftree}
\hskip -0.5pt
\AxiomC{$M = N : A$}
\AxiomC{$N = P : A$}
\LeftLabel{$(\tau)$ \ }
\BinaryInfC{$M = P : A$}
\end{bprooftree}

\bigskip

\noindent
\begin{bprooftree}
\hskip -0.5pt
\AxiomC{$M: (\Pi x : A)B$}
\LeftLabel{$(\eta)$ \ }
\RightLabel {$(x \notin FV(M))$}
\UnaryInfC{$(\lambda x.Mx) = M: (\Pi x : A)B$}
\end{bprooftree}

\bigskip

 Aside from these axioms, we can also apply a change of bounded variables. To simplify the notation, if not previously mentioned, consider that each computational path that appears in this work is considered modulo change of bounded variables.

\begin{definition}
Let $a$ and $b$ be elements of a type $A$. Then, a \emph{computational path} $s$ from $a$ to $b$ is a composition of applications of the inference rules of the equality theory of type theory. We denote that by $a =_{s} b$.
\end{definition}

The composition is given by applications of the transitive rule. We can better understand this fact by analyzing an example.  consider the construction of the path between $(\lambda y.yx)(\lambda w.zw) $ and $zx$. Starting from $(\lambda y.yx)(\lambda w.zw)$, we have that $(\lambda y.yx)(\lambda w.zw) \rhd_{\eta} (\lambda y.yx)z$. Then, we have the path $\eta((\lambda y.yx)(\lambda w.zw), (\lambda y.yx)z)$. From $(\lambda y.yx)z$, we have that $(\lambda y.yx)z \rhd_{\beta} zx$. Therefore, we have a path $\beta((\lambda y.yx)z,zx)$. We now need to compose these paths. To do this, we can apply the transitivity to obtain the desired path between  $(\lambda y.yx)(\lambda w.zw)$ and $zx$. Hence, the equality is established by the path $\tau(\eta((\lambda y.yx)(\lambda w.zw), (\lambda y.yx)z),\beta((\lambda y.yx)z,zx))$. So, in the end, a computational path is a term built out of the composition of (definitional) equality identifiers. A notion of canonical (path) identifier will be defined, and a set of reduction rules for (path) terms will be given explicitly.

In our recent works, we argue that computational paths are in fact the formal inhabitants of the identity type, since they can be seen as a `derivation-term' of a proof of propositional equality. Since a path $a =_{s} b : A$ establishes the equality between $a$ and $b$, then we can think of $s(a,b)$ as being an object of $Id_{A}(a,b)$. One can check \cite{Ruy1, Art1en,Ruy5} for more details.

\section{Reductions between computational paths}

In the previous section, we showed that computational paths establishes the equality between two terms of the same type. In this section, our objective is to investigate if two seemingly different computational paths could be considered equals, i.e., if there is the possibility of establishing a path between paths. Consider the simple case of a reflexive path $\rho$ that establishes that $a =_{\rho} a$. If we apply the axiom of symmetry $\sigma$, we obtain the path $a =_{\sigma(\rho)} a$. Since we applied the axiom $\sigma$ to the trivial case of reflexivity, we obtained a path that is just a redundant form of the reflexivity. In that case, the paths $\rho$ and $\sigma(\rho)$ should be considered equals, i.e., it should be possible to reduce $\sigma(\rho)$ to just $\rho$. The same thing happens if, given any path $a =_{s} b$, we apply the symmetry two times in succession, obtaining the path $a =_{\sigma(\sigma(s))} b$. Inverting a path two times is equivalent as just using the original path, i.e., $\sigma(\sigma(s))$ should be reduced to $s$.

As one could notice, the aforementioned redundancies have been originated from combination of applications of the axioms of the equality theory for type theory (as instance, the redundancy of our first example has been originated by the combination of applications of reflexivity and symmetry). In the previous section, using the example of the product type $\Pi$, we showed that we have a total of 8 equality axioms. Since we can combine these axioms in different situations, one should expect that the number of possible redundancies might be high. Fortunately, all these redundancies and reductions have already been mapped by \cite{Anjo1} and further developed (and given identifiers) by \cite{Ruy1}. That way, \cite{Anjo1} created a system that maps all these reductions, called $LND_{EQ}-TRS$, with a total of 39 reduction rules. To construct the fundamental groupoid, 7 rules will be essential. They are as follows:

\begin{itemize}

\item Reductions involving $\sigma$ and $\rho$:

$\sigma(\rho) \rhd_{sr} \rho$ \quad $\sigma(\sigma(r)) \rhd_{ss} r$

\item Reductions involving $\tau$

$\tau(r,\sigma(r)) \rhd_{tr} \rho$ \quad $\tau(\sigma(r),r) \rhd_{tsr} \rho$

$\tau(r,\rho) \rhd_{trr} r$ \quad $\tau(\rho,r) \rhd_{tlr} r$

\item Reductions involving $\tau$ and $\tau$

$\tau(\tau(t,r),s) \rhd_{tt} \tau(t, \tau(r,s))$

\end{itemize}

As one can see, an identifier is given to each reduction rule. For example, our initial examples are resolved by applications of the rule $sr$ and the rule $ss$ respectively.

It is important to note that the system $LND_{EQ}-TRS$, is terminating and confluent, as proved in \cite{Anjo1,Ruy2,Ruy3,RuyAnjolinaLivro}.

We call each reduction rule as a rewrite rule and usually abbreviated to $rw$-rule. Associated to an $rw$-rule, we have the following definitions:

\begin{definition}
Let $s$ and $t$ be computational paths. We say that $s \rhd_{1rw} t$ (read as: $a$ $rw$-contracts to $b$) iff we can obtain $t$ from $s$ by an application of only one $rw$-rule. If $s$ can be reduced to $t$ by finite number of $rw$-contractions, then we say that $s \rhd_{rw} t$ (read as $s$ $rw$-reduces to $t$).

\end{definition}

\begin{definition}
Let $s$ and $t$ be computational paths. We say that $s =_{rw} t$ (read as: $s$ is $rw$-equal to $t$) iff $t$ can be obtained from $s$ by a finite (perhaps empty) series of $rw$-contractions and reversed $rw$-contractions. In other words, $s =_{rw} t$ iff there exists a sequence $R_{0},....,R_{n}$, with $n \geq 0$, such that

\centering $(\forall i \leq n - 1) (R_{i}\rhd_{1rw} R_{i+1}$ or $R_{i+1} \rhd_{1rw} R_{i})$

\centering  $R_{0} \equiv s$, \quad $R_{n} \equiv t$
\end{definition}

By definition, the $rw$-equality is a class of equivalence, since it was defined as the symmetric, reflexive and transitive closure of $rw$-reduction. Since an $rw$-equality is given by a sequence of $rw$-reductions, we call this sequence an $rw$-sequence.

\subsection{Reduction between $rw$-equalities}

We showed that it is possible to think of reductions between paths and that these reductions were originated from redundancies caused by the combination of the equality axioms. Since $rw$-equalities are transitive, symmetric and reflexive, it is possible to think of reductions between two $rw$-sequences. For example, if $\theta$ is an $rw$-sequence, we can think of $\sigma(\sigma(\theta))$, which should be equivalent to $\theta$. In fact, knowing that the $ss$ rule resolves this kind of redundancy for computational paths, we can think of a $ss_{2}$ rule that resolves this redundancy for $rw$-sequences. In fact, since our rules that we showed previously only involve transitivity, reflexivity and symmetry, and that $rw$-equality has these properties, then we get an $rw$-equality version of each rule, only adding the label $2$ to make clear that we are working with $rw$-sequences and $rw$-equality. In fact, analogously to $rw$-equality, we can think of a $rw_{2}$-equality, which has the same definition, with the difference that we are working with reductions between $rw$-sequences. For that reason, $rw_{2}$-equality is also an equivalence class.

Since $rw_{2}$-equality forms a new system of rules, it is possible to think of rules that are specific to $rw_{2}$-equality. In fact, there is a specific rule that will play an essential role in the construction of our higher fundamental groupoid $\Pi_{2}(A)$. Consider the example of the path $\tau(s,t)$, such that $s \rhd_{1rw} s'$ and $t \rhd_{1rw} t'$. There are two possible $rw$-sequences: $\tau(s,t) \rhd_{1rw} \tau(s',t) \rhd_{1rw} \tau(s',t')$ and $\tau(s,t) \rhd_{1rw} \tau(s,t') \rhd_{1rw} \tau(s',t')$. As one can see, the only difference between these two $rw$-sequences are the choice of which term to develop first. Since $s$ and $t$ followed the same reductions (i.e., $t$ reduced to $t'$ in both cases and $s$ to $s'$), then these two $rw$-sequences are two ways of expressing the same result. For that reason, there should be a rule that establishes the $rw_{2}$-equality of these $rw$-sequences. We call this rule the independence of choice inside the transitivity, and it is expressed by $cd_{2}$. Then, we have that  $(\tau(s,t) \rhd_{1rw} \tau(s',t) \rhd_{1rw} \tau(s',t')) =_{rw_{2cd}} (\tau(s,t) \rhd_{1rw} \tau(s,t') \rhd_{1rw} \tau(s',t'))$. As one could see, this rule only makes sense in the context of $rw$-sequences and, for this reason, it is specific to $rw_{2}$-equality and $rw_{2}$-reduction.

\section{The fundamental groupoid of a type: $\Pi(A)$}

In this section, our objective is to show that the computational paths are capable of constructing the fundamental groupoid of a type $A$. The first thing that we should be aware is the idea of equality holding in a weak sense. In the homotopical interpretation, the idea of weak is related to a structure which the equalities only hold up to homotopy. In this work, we will use the term weak to characterize a structure which equalities only hold up to $rw$-equality or $rw_{2}$-equality. If the equalities hold 'on the nose', instead of being weak, we say that the structure is strict.

First, let's recall some basic concepts of category theory \cite{Steve2}:

\begin{definition}

Let $f: A \rightarrow B$ be an arrow of any category. $f$ is called an isomorphism if there exists a $g: B \rightarrow A$ such that $g \circ f  = 1_{A}$ and $f \circ g = 1_{B}$. $g$ is called the inverse of $f$ and can be written as $f^{-1}$.

\end{definition}

\begin{definition}
A groupoid is a category in which every arrow is an isomorphism.
\end{definition}

We can now think of a structure $\Pi(A)$ in which the objects are the objects $a : A$ and the morphisms are the computational paths between these objects:

\begin{proposition}
$\Pi(A)$ is a weak groupoid.
\end{proposition}

\begin{proof}

First, we need to prove that $\Pi(A)$ is a weak category. To do that, we need to define composition of morphisms and the identity arrow. We already know that composition of paths $s \circ t$ is given by application of the transitivity, i.e., $s \circ t = \tau(t,s)$. The identity morphism of an object $a$ is given by the reflexive path $a =_{\rho_{a}} a$. We need now to check the associative and identity laws. The associativity equation holds weakly, we just need to use the $tt$ rule:

\begin{center}
$\tau(\tau(s,r),t) =_{rw_{tt}} \tau(s,\tau(r,t))$.
\end{center}

Using rules $tlr$ and $trr$, we show that the identity laws also hold weakly:

\begin{center}
$s \circ 1_{a} = s \circ \rho_{a} = \tau(\rho_{a},s) =_{rw_{tlr}} s$

$1_{b} \circ s = \rho_{b} \circ s = \tau(s,\rho_{b}) =_{rw_{trr}} s$

\end{center}

With these conditions satisfied, we conclude that $\Pi(A)$ is indeed a weak category. To conclude our proof, we need to show that it also is a weak groupoid. To do that, we need to show that every computational path $s$ has a inverse computational path $s'$. Finding $s'$ is easy, just put $s' = \sigma(s)$. To show that the equalities of the isomorphism hold weakly, we use rules $tr$ and $tsr$:

\begin{center}

$s \circ s' = s \circ \sigma(s) = \tau(\sigma(s),s) =_{rw_{tsr}} \rho_{b}$

$s' \circ s = \sigma(s) \circ s = \tau(s,\sigma(s)) =_{rw_{tr}} \rho_{a}$

With that, we conclude that $\Pi(A)$ is a weak groupoid.

\end{center}
\end{proof}

Since the groupoid $\Pi(A)$ is constructed from a type $A$ and its objects and equalities between these objects (given by computational paths), we define $\Pi(A)$ as the fundamental groupoid of type $A$.

As proved by \cite{Streicher2}, the groupoid model of a type proves that the uniqueness of identity proofs is not derivable in the syntax of type theory. The construction of the fundamental groupoid of a type based on computational paths makes clear that the addition of computational paths to the syntax of type theory will still refute the uniqueness of identity proofs. In this sense, one can interpret this result as the fact that there may be two computational paths $s$ and $t$ between objects $a : A$ and $b : A$ such that $s$ and $t$ are not $rw$-equal.

\subsection{The fundamental 2-groupoid $\Pi_{2}(A)$ of a type}

Using the homotopical interpretation, one can construct a weak $2$-groupoid of a space $A$. This $2$-groupoid, called the fundamental $2$-groupoid of the space $A$, can be constructed by taking $\Pi(A)$ and adding a groupoid structure to each pair of objects of $\Pi(A)$. In this sense, for each pair of objects $a, b \in \Pi(A)$, we have a groupoid $\Pi(a,b)$, in which the objects are homotopical paths between $a$ and $b$ and morphisms are classes of homotopies between these paths. It is a well established fact that $\Pi(A)$ together with the structures $\Pi(a,b)$ forms a weak $2-groupoid$, known as $\Pi_{2}(A)$ \cite{Tom}.

The objective of this subsection is to show that it is possible to construct, using computational paths, the fundamental $2$-groupoid $\Pi_{2}(A)$ of a type $A$. To construct $\Pi_{2}(A)$, we can think of the structure $\Pi(A)$ with the addition of substructures $\Pi(A)(a,b)$ between each pair of objects of $\Pi(A)$. In $\Pi(A)(a,b)$, objects are computational paths between $a : A$ and $b : A$ and morphisms are $rw$-equalities modulo $rw_{2}$-equality. $\Pi(a,b)$ is a strict groupoid. The proof of this fact is analogous to \textbf{proposition 4.1}. The difference is that instead of using $rw$-rules, this proof uses the equivalent $rw_{2}$-rule. Another difference is that the equations hold strictly, since the morphisms are modulo $rw_{2}$-equality. Now we can prove the following proposition:

\begin{proposition}

Given any type $A$, $\Pi_{2}(A)$ is a weak $2$-groupoid.

\end{proposition}

\begin{proof}

First of all, we can represent $\Pi_{2}(A)$ graphically as follows:

\bigskip

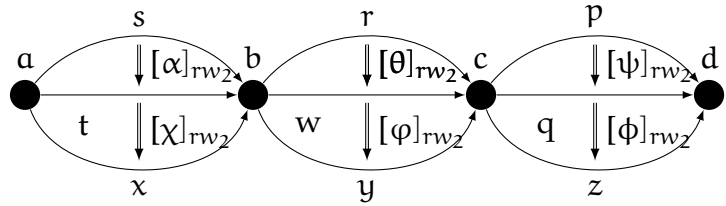
\begin{center}

\begin{tikzpicture}[>=latex,
    every node/.style={auto},
    arrowstyle/.style={double,->,shorten <=3pt,shorten >=3pt},
    mydot/.style={circle,fill}]

    \coordinate[mydot,label=above:$a$](a)  at (0,0);
    \coordinate[mydot,label=above:$b$](b) at (3,0);
    \coordinate[mydot,label=above:$c$](c) at (6,0);
    \coordinate[mydot,label=above:$d$](d) at (9,0);
    \draw[->]   (a) to[bend left=50]  coordinate (s) node[]{$s$}          (b);
    \draw[->]   (a) to coordinate (t) node[label={[label distance=-1cm]43:t}] {}    (b);
    \draw[->]   (a) to[bend right=70]  coordinate (x) node[,swap]{$x$}          (b);
    \draw[->]   (b) to[bend left=50]  coordinate (r) node[]{$r$}          (c);
    \draw[->]   (b) to coordinate (w) node[label={[label distance=-1cm]40:w}] {}    (c);
    \draw[->]   (b) to[bend right=70]  coordinate (y) node[,swap]{$y$}          (c);
    \draw[->]   (c) to[bend left=50]  coordinate (p) node[]{$p$}          (d);
    \draw[->]   (c) to coordinate (q) node[label={[label distance=-1cm]50:q}] {}    (d);
       \draw[->]   (c) to[bend right=70]  coordinate (z) node[,swap]{$z$}          (d);

    \draw[arrowstyle] (s) to node[right]{[$\alpha]_{rw_{2}}$} (t);
     \draw[arrowstyle] (t) to node[right]{[$\chi]_{rw_{2}}$} (x);
    \draw[arrowstyle] (r) to node[right]{[$\theta]_{rw_{2}}$} (w);
 \draw[arrowstyle] (w) to node[right]{[$\varphi]_{rw_{2}}$} (y);
\draw[arrowstyle] (r) to node[right]{[$\theta]_{rw_{2}}$} (w);
   \draw[arrowstyle] (p) to node[right]{[$\psi]_{rw_{2}}$} (q);
\draw[arrowstyle] (q) to node[right]{[$\phi]_{rw_{2}}$} (z);
\end{tikzpicture}	

\end{center}

We already know that $\Pi(A)$ is a weak groupoid and that $\Pi(A)(a,b)$ is a strict groupoid. To finish the proof that $\Pi_{2}(A)$ is a weak $2$-groupoid, we need to show that $\Pi_{2}(A)$ is a bicategory.

As one can check in \cite{Tom}, the proof that a structure is a bicategory follows a considerable number of steps. To achieve this proof, one of the first steps is to define how a $2$-morphism (as we already know, in $\Pi_{2}(A)$ $2$-morphisms are $rw$-equalities modulo $rw_{2}$-equality) can be composed horizontally. Given an $rw$-sequence $[\alpha]_{rw_{2}}: [s = \alpha_{1}, \alpha_{2}, ... , \alpha_{n} = t]_{rw_{2}}$ and an $rw$-sequence $[\theta]_{rw_{2}}: [r = \theta_{1}, \theta_{2}, ... , \theta_{m} = w]_{rw_{2}}$, we define the horizontal composition $[\alpha]_{rw_{2}} \circ_{h} [\theta]_{rw_{2}}$ as the following $rw$-sequence:

$[\alpha]_{rw_{2}} \circ_{h} [\theta]_{rw_{2}} = [\tau(s,r) = \tau(\alpha_{1},\theta_{1}), \tau(\alpha_{2}, \theta_{1}), ... \tau(\alpha_{n} = t,\theta_{1}), \tau(\alpha_{n}, \theta_{2}) ... ,\tau(\alpha_{n}, \theta_{m}) = \tau(t,w)]_{rw_{2}}$

With the horizontal composition well defined, we need to check if it is associative and respect the identity laws. For the associativity, we need to show that there is a natural isomorphism $assoc$ between $(([\psi]_{rw_{2}} \circ_{h} [\theta]_{rw_{2}}) \circ_{h} [\alpha]_{rw_{2}})$ and  $([\psi]_{rw_{2}} \circ_{h} ([\theta]_{rw_{2}} \circ_{h} [\alpha]_{rw_{2}}))$. To prove that, we will use the well established fact that a natural transformation is a natural isomorphism iff every component is an isomorphism \cite{Steve2}. That way, we only need to prove that there is an isomorphism between every component of $(([\psi]_{rw_{2}} \circ_{h} [\theta]_{rw_{2}}) \circ_{h} [\alpha]_{rw_{2}})$ and  $([\psi]_{rw_{2}} \circ_{h} ([\theta]_{rw_{2}} \circ_{h} [\alpha]_{rw_{2}}))$. In fact, since we are working in a groupoid structure, we just need to check that there is a morphism between each component (since every morphism is an isomorphism in a groupoid). To do this, we can get a generic component $\tau(\alpha_{x}, \tau(\theta_{y}, \psi_{z}))$ of  $(([\psi]_{rw_{2}} \circ_{h} [\theta]_{rw_{2}}) \circ_{h} [\alpha]_{rw_{2}})$ ($x, y, z$ being suitable naturals that respect the order of the horizontal composition) and show that it has a morphism to the equivalent component $\tau(\tau(\alpha_{x},\theta_{y}), \psi_{z})$ of  $([\psi]_{rw_{2}} \circ_{h} ([\theta]_{rw_{2}} \circ_{h} [\alpha]_{rw_{2}}))$. This morphism is established by $\sigma(tt)$ as follows:

\begin{center}
$\tau(\alpha_{x}, \tau(\theta_{y}, \psi_{z})) =_{rw_{\sigma(tt)}} \tau(\tau(\alpha_{x},\theta_{y}), \psi_{z})$
\end{center}

To prove the identity law, we use a similar idea. We need to check the natural isomorphism  $r^{*}_{s}$ between  $[\alpha]_{rw_{2}} \circ_{h} [\rho_{\rho_{a}}]_{rw_{2}}$  and $[\alpha]_{rw_{2}}$. To do that, we prove find a morphism between components $\tau(\rho_{\rho_{a}}, \alpha_{y}))$ and $\alpha_{y}$:

\begin{center}
$\tau(\rho_{\rho_{a}}, \alpha_{y}) =_{rw_{tlr}} \alpha_{y}$.
\end{center}

We now need to check the second identity law, i.e., the natural isomorphism $l^{*}_{s}$ between  $( [\rho_{\rho_{b}}]_{rw_{2}} \circ_{h} [\alpha]_{rw_{2}})$ and $[\alpha]_{rw_{2}}$. Taking components $\tau(\alpha_{y}, \rho_{\rho_{a}})$ and $\alpha_{y}$:

\begin{center}
$\tau(\alpha_{y}, \rho_{\rho_{a}}) =_{rw_{trr}} \alpha_{y}$.
\end{center}

Also associated to the horizontal composition, we need now to check the interchange law, i.e., we need to check:

\begin{center}
$([\varphi]_{rw_{2}} \circ [\theta]_{rw_{2}}) \circ_{h} ([\chi]_{rw_{2}} \circ [\alpha]_{rw_{2}}) = ([\varphi]_{rw_{2}} \circ_{h} [\chi]_{rw_{2}}) \circ ([\theta]_{rw_{2}} \circ_{h} [\alpha]_{rw_{2}})$
\end{center}

From  $(([\varphi]_{rw_{2}} \circ [\theta]_{rw_{2}}) \circ_{h} ([\chi]_{rw_{2}} \circ [\alpha]_{rw_{2}}))$, we have:

\begin{center}
 $(([\varphi]_{rw_{2}} \circ [\theta]_{rw_{2}}) \circ_{h} ([\chi]_{rw_{2}} \circ [\alpha]_{rw_{2}})) =$

 $[\tau(\theta, \varphi)]_{rw_{2}} \circ_{h} [\tau(\alpha, \chi)]_{rw_{2}}=$

$[\theta_{1}, ..., \theta_{n} = \varphi_{1}, ..., \varphi_{n'}]_{rw_{2}} \circ_{h} [(\alpha_{1}, ..., \alpha_{m} = \chi_{1}, ... \chi_{m'}]_{rw_{2}} = $

$[\tau(\alpha_{1}, \theta_{1}),..., \tau(\alpha_{m} = \chi_{1}, \theta_{1}), ..., \tau(\chi_{n},\theta_{1}), ..., \tau(\chi_{n},\theta_{m'} = \varphi_{1}),..., \tau(\chi_{n},\varphi_{n'})]_{rw_{2}}$
\end{center}

From   $(([\varphi]_{rw_{2}} \circ_{h} [\chi]_{rw_{2}}) \circ ([\theta]_{rw_{2}} \circ_{h} [\alpha]_{rw_{2}})))$:

\begin{center}
 $(([\varphi]_{rw_{2}} \circ_{h} [\chi]_{rw_{2}}) \circ ([\theta]_{rw_{2}} \circ_{h} [\alpha]_{rw_{2}}))(r \circ s) =$

 $([\tau(\chi_{1}, \varphi_{1}),...,\tau(\chi_{n},\varphi_{1}),...,\tau(\chi_{n},\varphi_{n'})]_{rw_{2}} \circ [\tau(\alpha_{1},\theta_{1}),...,\tau(\alpha_{m},\theta_{1}),...,\tau(\alpha_{m},\theta_{m'})]_{rw_{2}} = $

$[\tau(\alpha_{1},\theta_{1}),...,\tau(\alpha_{m},\theta_{1}),...,\tau(\alpha_{m},\theta_{m'}),...,\tau(\chi_{1}, \varphi_{1}),...,$ $\tau(\chi_{n},\varphi_{1}),...,\tau(\chi_{n},\varphi_{n'})]_{rw_{2}}$

\end{center}

As one can see, the sole difference between the two results is the order that the internal $rw$-sequences have been developed in each transitivity. This is a perfect case to use the independence of choice inside the transitivity, i.e., the $rw_{2}$-rule known as $cd_{2}$. By one application of $cd_{2}$, we conclude that $[\tau(\alpha_{1}, \theta_{1}),..., \tau(\alpha_{m} = \chi_{1}, \theta_{1}), ..., \tau(\chi_{n},\theta_{1}), ..., \tau(\chi_{n},\theta_{m'} = \varphi_{1}),..., \tau(\chi_{n},\varphi_{n'})]_{rw_{2}}$ $= [\tau(\alpha_{1},\theta_{1}),...,\tau(\alpha_{m},\theta_{1}),...,\tau(\alpha_{m},\theta_{m'})$ $,...,\tau(\chi_{1}, \varphi_{1}),...,\tau(\chi_{n},\varphi_{1}),...,\tau(\chi_{n},\varphi_{n'})]_{rw_{2}}$.

To end our proof, since we are working with a weak structure, we need to check the coherence laws for a bicategory. The coherence laws are given by Mac Lane's pentagon and triangle \cite{Tom}:

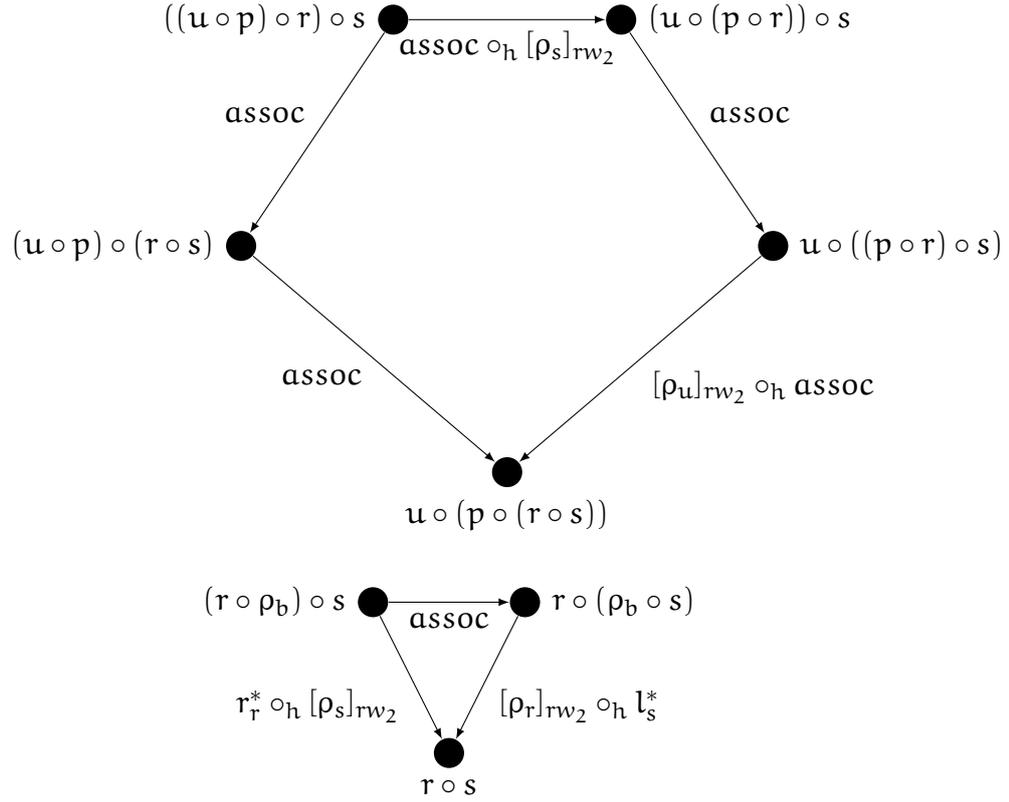
\begin{center}
\begin{tikzpicture}[>=latex,
    every node/.style={auto},
    arrowstyle/.style={double,->,shorten <=3pt,shorten >=3pt},
    mydot/.style={circle,fill}]

    \coordinate[mydot,label=left:$((u\circ p) \circ r) \circ s$](a)  at (2,0);
    \coordinate[mydot,label=left:$(u \circ p) \circ (r \circ s)$](b) at (0,-3);
    \coordinate[mydot,label=right:$(u \circ(p \circ r)) \circ s$](c) at (5,0);
    \coordinate[mydot,label=right:$u \circ((p \circ r) \circ s) $](d) at (7,-3);
     \coordinate[mydot,label=below:$u \circ (p \circ (r \circ s))$](e) at (3.5,-6);

    \draw[->]   (a) to coordinate (s) node[,swap] {$assoc$}    (b);
    \draw[->]   (a) to coordinate (t) node[,swap] {$assoc \circ_{h} [\rho_{s}]_{rw_{2}}$}    (c);
    \draw[->]   (c) to coordinate (x) node[] {$assoc$}    (d);
     \draw[->]   (b) to coordinate (y) node[,swap] {$assoc$}    (e);
     \draw[->]   (d) to coordinate (z) node[] {$[\rho_{u}]_{rw_{2}} \circ_{h} assoc$}    (e);
\end{tikzpicture}

\bigskip

\begin{tikzpicture}[>=latex,
    every node/.style={auto},
    arrowstyle/.style={double,->,shorten <=3pt,shorten >=3pt},
    mydot/.style={circle,fill}]

  \coordinate[mydot,label=left:$(r \circ \rho_{b}) \circ s$](a)  at (0,0);
    \coordinate[mydot,label=right:$r \circ (\rho_{b} \circ s)$](b) at (2,0);
    \coordinate[mydot,label=below:$r \circ s$](c) at (1,-2);

        \draw[->]   (a) to coordinate (s) node[,swap] {$assoc$}    (b);
       \draw[->]   (a) to coordinate (t) node[,swap] {$r^{*}_{r} \circ_{h} [\rho_{s}]_{rw_{2}}$}    (c);
      \draw[->]   (b) to coordinate (x) node[] {$[\rho_{r}]_{rw_{2}} \circ_{h} l^{*}_{s}$}    (c);
\end{tikzpicture}
\end{center}

We need to show that the diagrams above commute. The equations are straightforward. Starting with the pentagon, we can start from $((u\circ p) \circ r) \circ s = \tau(s,\tau(r,\tau(p,u)))$ and go to the right of the diagram:

\begin{center}
$(assoc \circ_{h} [\rho_{s}]_{rw_{2}})(\tau(s,\tau(r,\tau(p,u)))) = \tau(s,assoc(\tau(r,\tau(p,u)))) =$

$\tau(s,\tau(\tau(r,p),u))$

$assoc(\tau(s,\tau(\tau(r,p),u))) = \tau(\tau(s,\tau(r,p)),u)$

$([\rho_{u}]_{rw_{2}} \circ_{h} assoc)(\tau(\tau(s,\tau(r,p)),u)) = \tau(assoc(\tau(s,\tau(r,p))), u) = $

$\tau(\tau(\tau(s,r),p)), u) = u \circ (p \circ (r \circ s))$
\end{center}

Starting from the same $((u\circ p) \circ r) \circ s = \tau(s,\tau(r,\tau(p,u)))$ and going bottom left:

\begin{center}
$assoc(\tau(s,\tau(r,\tau(p,u)))) = \tau(\tau(s,r),\tau(p,u))$

$assoc(\tau(\tau(s,r),\tau(p,u))) = \tau(\tau(\tau(s,r),p),u) = u \circ (p \circ (r \circ s))$
\end{center}

We conclude that the pentagon commutes. For the triangle, we can start with $((r \circ \rho_{b}) \circ s) = \tau(s,\tau(\rho_{b},r))$ and go to the right:

\begin{center}
$assoc(\tau(s,\tau(\rho_{b},r))) = \tau(\tau(s,\rho_{b}),r)$

$([\rho_{r}]_{rw_{2}} \circ_{h} l^{*}_{s})\tau(\tau(s,\rho_{b}),r) = \tau(l^{*}_{s}(\tau(s,\rho_{b})),r) = $

$\tau(s,r) = r \circ s$
\end{center}

Now starting with the same $((r \circ \rho_{b}) \circ s) = \tau(s,\tau(\rho_{b},r))$ and going to to bottom right:

\begin{center}
$(r^{*}_{r} \circ_{h} [\rho_{s}]_{rw_{2}})\tau(s,\tau(\rho_{b},r)) = \tau(s,r^{*}_{r}(\tau(\rho_{b},r))) =$

$\tau(s,r) = r \circ s$
\end{center}

The triangle also commutes. The coherence laws hold. With that, we finish the proof that $\Pi_{2}(A)$ is a weak $2$-groupoid.

\end{proof}

With this proof, we conclude that, using computational paths, it is possible to construct the fundamental $2$-groupoid of a type $A$. In fact, using the same reasoning that led us to come up with $rw_{2}$-equality, one could think of $rw_{3}$-equality, $rw_{4}$-equality, etc. That way, it would be possible to think of the possibility of even higher fundamental groupoids. If we continue this process to the infinite, we would think of a possible weak infinite groupoid $\Pi_{\infty}(A)$, also called weak $\omega$-groupoid. As a matter of fact, it has already been proved that it is possible to obtain such structure, as shown in \cite{lumsdaine1, Benno}. Our objective, in a future work, is to obtain similar results using the concept of computational paths. The same way that we used to construct the fundamental $2$-groupoid $\Pi_{2}(A)$, we could use computational paths and infinite levels of $rw$-equalities to construct $\Pi_{\infty}(A)$. The problem is that $\Pi_{\infty}(A)$ is a complex structure and to prove that it is a weak $\omega$-groupoid is not an easy task. For that reason, obtaining this construct will be the main focus of some of our future works.

\section{Conclusion}

The main objective of this work was to construct a groupoid model for a type using the concept of computational paths. Based on the idea that it is possible to introduce an entity known as computational paths to the syntax of type theory, we used this newly added entity to construct a groupoid model for a type $A$. We have showed that this approach is fundamentally different from the groupoid constructed using homotopical paths, since computational paths are de facto elements of the syntax of type theory, instead of being only semantical interpretations.

To achieve our results, we have showed that it is possible to think of reductions between two different computational paths. Based on these reductions, we have constructed a structure in which objects are elements of a type $A$ and morphisms are computational paths. We have showed that this structure is, in fact, a weak groupoid, and called it the fundamental groupoid of a type $A$. We have gone further, defining a higher structure and proving that this structure is also a groupoid. That way, using computational paths, we have obtained the fundamental $2$-groupoid of the type $A$. A rather natural way forward is to show that it is indeed possible to obtain even higher structures, with the main goal of using computational paths to obtaining a possible fundamental $\infty$-groupoid.

\bibliographystyle{ieeetr}
\bibliography{Biblio}

\end{document}